\documentclass[a4paper,UKenglish,cleveref, autoref, thm-restate]{lipics-v2021}

\pdfoutput=1 
\hideLIPIcs  

\usepackage{aliascnt, enumitem}
\newcommand{\R}{\mathbb{R}}
\newcommand{\E}{\mathbb{E}}
\newcommand{\Prob}{\mathbb{P}}
\newcommand{\F}{\mathcal{F}}

\newcommand{\SR}{{SecureR}}

\newtheorem{assumption}[theorem]{Assumption}

\crefname{theorem}{theorem}{theorems}
\Crefname{theorem}{Theorem}{Theorems}
\crefname{proposition}{proposition}{propositions}
\Crefname{proposition}{Proposition}{Propositions}
\crefname{lemma}{lemma}{lemmas}
\Crefname{lemma}{Lemma}{Lemmas}
\crefname{corollary}{corollary}{corollaries}
\Crefname{corollary}{Corollary}{Corollaries}
\crefname{definition}{definition}{definitions}
\Crefname{definition}{Definition}{Definitions}
\crefname{remark}{remark}{remarks}
\Crefname{remark}{Remark}{Remarks}
\crefname{assumption}{assumption}{assumptions}
\Crefname{assumption}{Assumption}{Assumptions}

\title{Reserve Depletion and Security Runway in Proof-of-Stake Systems} 

\author{Paolo Penna}{IOG,  Switzerland}{paolo.penna@iohk.io}{}{}
\author{Manvir Schneider}{Cardano Foundation, Switzerland}{manvir.schneider@cardanofoundation.org}{}{}

\authorrunning{M. Schneider and P. Penna} 

\Copyright{} 

\ccsdesc[100]{} 

\keywords{Blockchain, Proof-of-Stake, Reserve, Tokenomics, Equilibrium, Repeated Game} 

\category{} 

\relatedversion{} 

\nolinenumbers 

\EventEditors{}
\EventNoEds{1}
\EventLongTitle{}
\EventShortTitle{}
\EventAcronym{}
\EventYear{}
\EventDate{}
\EventLocation{}
\EventLogo{}
\SeriesVolume{}
\ArticleNo{}

\begin{document}

\maketitle

\begin{abstract}
Many proof-of-stake protocols finance validator rewards from two sources:
transaction fees and a finite reserve of native tokens. This creates a
dynamic hand-off problem. Early in the life of the system, fees may be too
small to fund the target level of security; later, fees may become sufficient.
The central question is whether the reserve provides enough runway for the
protocol to remain secure until this fee-only region is reached.
We study this problem in a discrete-time stochastic model of validator
participation. Token price and transaction demand fluctuate over time, while
validators choose participation strategically in each state. We solve the
validator entry game in closed form and derive an exact state-dependent reserve
threshold: for every token price and demand state, the threshold is the minimal
reserve stock necessary and sufficient to sustain a target security level.
This threshold separates three regions: infeasibility, reserve-dependent
security, and fee-only security.
The threshold turns reserve adequacy into a hitting-time problem. Security
fails when the reserve first falls below the state-dependent threshold, and a
successful hand-off occurs exactly when the fee-only region is reached before
that failure time. We derive conservative finite-horizon stress-test guarantees
that convert lower confidence bands for token price and demand into reserve
requirements, and under lognormal price-demand dynamics we obtain explicit
failure-probability and expected hand-off-time bounds. Finally, we extend the
model to forward-looking validators and derive the Markov participation
condition that captures how current participation affects future reserve-funded
rewards.
The main implication is that reserve policy should not be evaluated by nominal
depletion dates or steady-state reward ratios alone. A protocol can have a large
nominal reserve and still be close to security failure after adverse price or
demand shocks. Conversely, once demand crosses the fee-only threshold, the
reserve becomes redundant for security. This paper provides a tractable
equilibrium framework for stress-testing this transition.
\end{abstract}

\section{Introduction}

Proof-of-stake blockchains typically pay validators from two sources: current transaction fees and a reserve of previously minted or retained tokens. The reserve is meant to bridge the early phase of the system, when adoption is still low and fees alone do not support the amount of active stake needed for security. Cardano is a canonical example of such a design: its reward system combines transaction fees with a controlled drawdown of reserves \cite{cardanoMonetaryPolicyDocs}. Other examples include Avalanche \cite{avalanche}, Algorand \cite{algorand_two}, Ripple \cite{ripple}, and Hedera \cite{hedera}. The policy problem is therefore not simply whether fee revenue is positive in the long run. The relevant question is whether the protocol can remain secure \emph{along the entire transition path}.

That transition problem has two distinct sources of fragility. First, fee income is uncertain because transaction demand is uncertain. Second, even if the reserve is large in token units, its purchasing power is stochastic because validator costs are incurred in an external numeraire while the reserve is held in the native token. A reserve that looks ample at one token price can become inadequate after a price drawdown.

This creates a governance problem that is not visible from the reserve balance
alone. Suppose two protocols have the same number of reserve tokens. The first
has strong current demand and a high token price; the second faces weak demand
and a price drawdown. Their nominal reserves are identical, but their ability to fund validator participation is not. Conversely, a protocol with a declining reserve may be safe if fee demand has already grown enough to support the target security level without subsidies. The relevant object is therefore a state-contingent reserve
requirement.

The security-runway perspective is designed to capture exactly this issue. A protocol has runway at time $t$ if the reserve stock at that date is large enough, given the current state of demand and token price, to keep equilibrium security above the required minimum. Runway ends not when the reserve literally reaches zero, but when it first becomes too small to finance the security shortfall.

This paper studies the following fundamental question:
\begin{quote}
    \em{Can a blockchain protocol guarantee secure operation throughout the transition from
reserve-funded to fee-only validator rewards, and under what conditions does this transition fail?}
\end{quote}
We answer this question by deriving a reserve threshold that can be interpreted as a state-contingent runway criterion: given price, demand, and a target security level, it tells whether the current reserve is sufficient to sustain equilibrium security. This allows reserve
policy to be stress-tested state by state.

\paragraph*{Our Contribution}
We study the reserve-to-fee transition in a stochastic state model with two
exogenous state variables: token price and user demand for blockspace. Validators
choose participation strategically, so equilibrium security and fees respond to
the current reserve, price, and demand state. Our contributions are as follows.
\begin{enumerate}[label=(\roman*)]
\item We solve the validator participation game in closed form and prove
existence and uniqueness of a symmetric Nash equilibrium
(\Cref{thm:unique-sne}). The equilibrium shows that the security value of a
given reserve stock depends on the current token price.

\item We derive an exact state-dependent reserve threshold
(\Cref{thm:threshold}). For each price-demand state, the threshold is the
minimal reserve stock necessary and sufficient to sustain a target security
level.

\item We use this threshold to formulate the dynamic hand-off problem
(Section~\ref{sec:dynamic}). Security persists while the reserve remains above
the threshold, and hand-off succeeds when the fee-only region is reached before
failure (\Cref{prop:runway}, \Cref{thm:stress-test}, \Cref{cor:fee-only}).

\item We extend the baseline model to a forward-looking Markov environment
(Section~\ref{sec:Markov}). We prove finite-horizon Markov perfect equilibrium
existence and characterize how continuation values alter current validator
participation
(\Cref{thm:mpe-existence}, \Cref{prop:mpe-uniqueness,prop:dynamic-foc,prop:markov-runway}).

\item We specialize the runway analysis to geometric Brownian token price and
discrete-time lognormal demand (Section~\ref{sec:Brownian}). This yields
explicit finite-horizon failure-probability guarantees and bounds on hand-off
timing (\Cref{thm:parametric-failure}, \Cref{prop:parametric-handoff}).
\end{enumerate}

\paragraph*{Related Literature}
Our analysis is related to three strands of literature. First, the Ouroboros line of work studies the protocol and security foundations of proof-of-stake systems and provides the incentive-theoretic background for stake-based consensus \cite{kiayias2017ouroboros,david2018praos}. Second, economic analyses of proof-of-stake study equilibrium incentives, reward design, and the distributional effects of staking rewards \cite{saleh2021blockchain,fanti2019compounding,irresberger2023coin}. Third, the blockchain fee-market literature studies how congestion pricing and transaction fees support decentralized infrastructure \cite{huberman2021monopoly,ma2022transaction}.  Dynamics models for cryptocurrencies include \cite{kiayias2023would,decentralization_friction,CK2022,DBLP:conf/aft/KiayiasLP25}, while geometric Brownian motion is used, e.g., in \cite{cong2021tokenomics,cong2022token}. To the best of our knowledge, existing work has not isolated the reserve
hand-off problem studied here: a finite native-token reserve, strategic validator participation, stochastic token price and demand, and a state-dependent threshold for sustaining a target security level until fee-only operation becomes possible.
Earlier work models the effects of  fee-only rewards in Bitcoin: First, higher variance  induce ``forking'' strategies of miners competing for transactions with high tips \cite{BitcoinInstability}. Moreover, attackers can bribe honest participants offering tips for mining a forked chain and increase the chances of successfully performing double spending \cite{IncentiveBlockchainForks}. 

\paragraph*{Contribution Relative to the Literature}
The present paper differs from these strands by focusing on a dynamic reserve-funded transition problem. The object of interest is not simply equilibrium rewards at a point in time, nor fee formation in isolation, but the reserve threshold required to sustain security until fee funding becomes sufficient on its own. That is the role of the security-runway concept developed below.

Our model has some similarities with the validators side in \cite{kiayias2023would}: There validators decide the amount of blockspace (or security) to provide, and their utility depends on the per-unit fee of the system as equilibrium.  That paper, however, does not consider reserve constraints, and it models users demand differently. 

In our model both token prices and demand are exogenous quantities. In \cite{decentralization_friction,DBLP:conf/aft/KiayiasLP25} instead, the system value is determined by the protocol evolution (at some cost function) and the token price is an endogenous quantity which is part of the equilibrium (the price at which tokens are exchanged in a spot market involving validators and users). In this respect, our model exhibits greater flexibility and generality, while additionally incorporating the aforementioned reserve constraint, which is absent in \cite{decentralization_friction,DBLP:conf/aft/KiayiasLP25}.

\section{Model}\label{sec:model}

\renewcommand{\xi}{FeeMax}
\renewcommand{\beta}{Sec}
\newcommand{\price}{Price}
\newcommand{\fee}{f}

Time is discrete and indexed by $t=0,1,2,\dots$. All random variables are defined on a filtered probability space
$
(\Omega,\F,(\F_t)_{t\ge 0},\Prob).
$
The protocol state at time $t$ is given by a tripel $ 
X_t=(R_t,\price _t,\xi_t),
$
where:
\begin{itemize}
\item $R_t \in \R_+$ is the reserve stock measured in native-token units.
\item $\price _t \in \R_{++}$ is the token price in an external numeraire. Validator costs are paid in this numeraire.
\item $\xi_t \in \R_{++}$ is a demand state governing willingness to pay for block space. Intuitively, this is the largest fee that at least some user is willing to pay (any higher fee results in zero demand). 
\end{itemize}
The policy parameters are fixed ex ante:
\begin{align*}
\theta \in [0,1), &&  \rho \in (0,1),
\end{align*}
where $\theta$ is the share of fee revenue diverted into the reserve and $\rho$ is the fraction of the reserve paid out each period. Token price and demand are stochastic, exogenous, and positive quantities: 

\begin{assumption}[State dynamics]
\label{ass:state}
The token price process $(\price _t)_{t\ge 0}$ and the demand process $(\xi_t)_{t\ge 0}$ are strictly positive and adapted to the filtration $(\F_t)_{t\ge 0}$.
\end{assumption}

\subsection{Game Form and Economic Scope}

The players are the $N$ validators. Their date-$t$ actions are participation levels, interpreted as active stake or security supply. Users are not modeled as strategic players; instead, user behavior is summarized by an inverse-demand schedule for block space. Likewise, token-market trading is not modeled strategically; it enters through the stochastic price process $\price _t$. This separation isolates the security-runway question: how the reserve interacts with validator incentives when current fee conditions and the external token price fluctuate over time.

This modeling choice has two implications that should be kept distinct throughout the paper. First, the token price $\price _t$ is externally given to validators at date $t$, so the model does not claim to explain token valuation. Second, the transaction fee $\fee_t$ is endogenous, because it is the market-clearing fee induced by validator participation and the current demand state $\xi_t$. Thus demand and token price are state variables, while validator participation and fees are equilibrium outcomes.

\subsection{A Reduced-Form Security Technology}

There are $N \ge 2$ validators. Validator $i$ chooses an amount of active stake, or security supply,
\[
a_{i,t} \in \R_+,
\]
and aggregate security and blockspace is
\begin{align}
\label{eq:blockspace}
s_t := \sum_{i=1}^N a_{i,t}.
\end{align}
The baseline model is symmetric across validators: every validator has the same feasible action set, the same cost function, and the same proportional reward-sharing rule. \footnote{The model therefore suppresses ex ante heterogeneity in stake, scale, or operating cost.}

The action \(a_{i,t}\) should be interpreted as validator \(i\)'s effective
supply of secure blockspace at date \(t\). It combines the validator's active
participation with the amount of transaction-processing capacity that the
validator makes available to users. Thus, a validator who is online but produces
empty or only partially filled blocks is represented as supplying a lower
effective \(a_{i,t}\). The aggregate \(s_t\) is therefore the quantity of
secure blockspace supplied in equilibrium. It enters the security side of the
model because larger effective participation raises the security level, and it
enters the demand side because it is the quantity of blockspace cleared at the
market fee. This convention keeps the model focused on the funding problem:
how much effective secure capacity validators are willing to supply when rewards
come from fees and reserve payouts.

Users are described by an \emph{inverse demand curve} expressed in terms of the external numeraire (for example, in dollars):
\begin{align}
    \label{eq:inv-dem}
P(s_t,\xi_t) := \xi_t - \beta \cdot s_t \ .    
\end{align}
Intuitively, given the available blockspace provided by the validators \eqref{eq:blockspace}, Equation~\ref{eq:inv-dem} gives the \emph{market-clearing} fee: 
\begin{itemize}
    \item Users demand equals the supplied blockspace $s$ if the protocol charges $P(s,\xi)$ dollars per unit of blockspace.~\footnote{Note that we are assuming each user to consume one unit of blockspace -- say one transaction. Hence, the demand corresponding to \eqref{eq:inv-dem} is of the form $D(p) = (\xi - p)/\beta$ for $p\in [0,\xi]$, and $D(p)=0$ for $p> \xi$.}
    \item The parameter $\beta>0$ determines the rate at which market-clearing fees decline as total blockspace supply increases (intuitively, the protocol must lower fees to accommodate higher user demand).
\end{itemize}
As already mentioned above, 
while $\xi_t$ is an exogenous parameter controlling users demand, blockspace $s_t$ is strategic and results from validators aiming at maximizing their own utilities.  If the supplied blockspace is $s_t$, the protocol posts a token-denominated fee
\begin{align}
\label{eq:token-fee}
\fee_t(s_t,\xi_t,\price_t) := \frac{P(s_t,\xi_t)}{\price _t} = \frac{\xi_t-\beta\cdot s_t}{\price _t}\ .
\end{align}
The resulting fee revenue, in external numeraire, is thus
\[
\price _t \cdot \fee_t \cdot s_t = (\xi_t-\beta\cdot s_t)\cdot s_t \ .
\]

The model restricts attention to action profiles that yield nonnegative clearing fees, reflecting the natural requirement that the system does not subsidize users for consuming blockspace. Since $P(s_t,\xi_t)\ge 0$ requires $s_t\le \xi_t/\beta$, we impose the symmetric feasibility bound
\begin{align}\label{eq:symmetric-feasibility-bound}
a_{i,t} \in \left[0,\frac{\xi_t}{N\cdot \beta}\right].    
\end{align}
This condition implies that $0 \le s_t \le \xi_t / \beta$, ensuring that clearing fees are always nonnegative. Furthermore, in the context of symmetric equilibria—where all validators select an identical strategy—these two conditions are equivalent.

\subsection{Reserve Dynamics and Validator Rewards}

The reserve absorbs a fraction $\theta$ of current fee revenue and pays out a fraction $\rho$ of the reserve each period. Because the reserve is held in tokens, its law of motion is
\begin{equation}
\label{eq:reserve-law}
R_{t+1} = (1-\rho)\cdot R_t + \theta\cdot  \frac{(\xi_t-\beta\cdot s_t)\cdot s_t}{\price _t}.
\end{equation}

The total reward pool paid to validators in period $t$, measured in the external numeraire, is
\begin{equation}
\label{eq:reward-pool}
W_t := (1-\theta)\cdot (\xi_t-\beta\cdot s_t)\cdot s_t + \rho \cdot \price_t \cdot R_t.
\end{equation}
The first term is the portion of fees not diverted to the reserve; the second term is the external value of the reserve payout.

Each validator has quadratic operating cost
\begin{align}\label{eq:val-costs}
\frac{\kappa}{2}a_{i,t}^2,
&&  \kappa>0 \ .
\end{align}
The quadratic form is the standard way to encode increasing marginal cost of active stake and guarantees an interior first-order condition whenever the equilibrium is not constrained by the capacity bound.

\subsection{Stage-Game Equilibrium}\label{sec:baseline}
In this section, we consider a single-stage game which is fully specified by the current state: The current reserve $R_t$, token price $\price_t$, and demand $\xi_t$. As we focus only on a generic stage $t$, we suppress $t$ and consider a generic  state $x=(R,\price,\xi)$. The stage game at state $x$ is played only by the $N$ validators. A pure action for validator $i$ is a feasible action $a_i$ satisfying \eqref{eq:symmetric-feasibility-bound}, that is, $a_i\in [0,\xi/(N\cdot \beta)]$. Actions are chosen simultaneously. If the aggregate action profile is $a=(a_1,\dots,a_N)$ and $s=\sum_i a_i>0$, validator $i$ receives its pro rata share of the reward pool \eqref{eq:reward-pool} given the current state and action profile, 
$$W=W(a;x):= \frac{a_i}{s}\cdot \left((1-\theta)\cdot (\xi-\beta\cdot s)\cdot s + \rho \cdot \price \cdot R\right).$$ The resulting utility (reward minus incurred cost) is thus 
\begin{align}
\label{eq:utility}
u_i(a;x) := \frac{a_i}{s}  \cdot W  - \frac{\kappa}{2}a_i^2 =  \frac{a_i}{s}\left((1-\theta)\cdot (\xi-\beta\cdot s)\cdot s + \rho \cdot \price \cdot R\right) - \frac{\kappa}{2}a_i^2.    
\end{align}
When all validators choose zero, we set $u_i(0,\dots,0;x)=0$.

\begin{definition}[Symmetric Nash equilibrium]
\label{def:sne}
Fix a state $x=(R,\price,\xi)$. A \emph{symmetric Nash equilibrium} is a profile $a^*(x)=(a^*(x),\dots,a^*(x))$ such that
\[
u_i\big(a^*(x),\dots,a^*(x);x\big)
\ge
u_i\big(a_i,a_{-i}^*(x);x\big)
\]
for every validator $i$ and every feasible deviation $a_i \in [0,\xi/(N\beta)]$, where $a_i,a^*_{-i}$ denotes the vector obtained by replacing the $i^{th}$ entry in $a^*$ with $a_i$.

The induced equilibrium security, fee, and next-period reserve are given by:
\begin{align*}
s^*(x)= & N\cdot a^*(x), \qquad \fee^*(x)\stackrel{\eqref{eq:token-fee}}{=}\frac{\xi-\beta\cdot s^*(x)}{\price}, \\ R^+(x)=& (1-\rho)\cdot R+\theta \cdot f^*(x) \cdot s^*(x) =  (1-\rho)\cdot R+\theta \cdot \frac{(\xi-\beta\cdot s^*(x))\cdot s^*(x)}{\price}.
\end{align*}
\end{definition}

Symmetric equilibria arise naturally in the context of blockchain systems (see, e.g., \cite{decentralization_friction,DBLP:conf/aft/KiayiasLP25}), as they constitute the intended design objective of many protocols: they correspond to a regime of maximal decentralization \cite{DBLP:conf/acns/OvezikKMWK25,motepalli2025decentralization}.

The definition above is for a generic state and yields a state by state evolution in the natural way:  Given the current state $x_t = (R_t,\price_t, \xi_t)$,  the reserve at the next state is $R_{t+1} = R^+(x_t)$, and therefore the next state is   $x_{t+1} = (R_{t+1},\price_{t+1},\xi_{t+1})$. Dynamic incentives are encoded only through the reserve stock carried to the next date; validators do not solve an intertemporal control problem. The analysis focuses on whether the current state provides enough funding to sustain current security, not on dynamic reputation effects.

Observe that the state is composed of three variables, of which two are exogenous: The reserve stock $R$ is an endogenous state variable inherited from past protocol outcomes. The token price $\price$ and the demand shifter $\xi$ are exogenous state variables from the perspective of the validator game. Conditional on a state $x$, equilibrium participation $a^*(x)$ and the clearing fee $\fee^*(x)$ are endogenous objects solved within the model.

\section{Static Equilibrium Analysis}

This section solves the stage game exactly. The main result is a closed-form symmetric Nash equilibrium, which then becomes the building block for the dynamic runway analysis.

\subsection{Existence and Uniqueness}

Define the constants
\begin{align}
\label{eq:A-B}
A_N := \frac{\kappa}{N} + \frac{N+1}{N}\cdot (1-\theta)\cdot \beta\ ,
&& 
B_N := \frac{N-1}{N}\cdot \rho\ .
\end{align}

The next theorem shows uniqueness of symmetric equilibria in the stage game. 

\begin{theorem}[Unique symmetric Nash equilibrium]
\label{thm:unique-sne}
For every state $x=(R,\price,\xi) \in \R_+ \times \R_{++}^2$, the stage game admits a unique symmetric Nash equilibrium. Its aggregate security level is
\begin{equation}
\label{eq:s-star}
s^*(x)
=
\min\left\{
\frac{(1-\theta)\cdot\xi + \sqrt{(1-\theta)^2\cdot \xi^2 + 4A_N B_N \cdot \price \cdot  R}}{2A_N},
\frac{\xi}{\beta}
\right\}.
\end{equation}
The corresponding equilibrium fee is
$
\fee^*(x)=\frac{\xi-\beta\cdot s^*(x)}{\price}
$. 
\end{theorem}

\begin{remark}
The equilibrium formula shows exactly how the reserve enters incentives. Current fee funding depends on demand $\xi$ and congestion $\beta$, while reserve funding enters only through the product $\price \cdot R$, the reserve's value in \emph{the external numeraire}. The same reserve stock therefore has different security consequences at different token prices. This is the main reason a deterministic reserve-only calculation can be misleading.
\end{remark}

Theorem~\ref{thm:unique-sne} implies that, in any symmetric equilibrium, the validators' utilities are given by the following expression.

\begin{corollary}\label{cor:utilities-equilibrium}
    For every state $x=(R,\price,\xi) \in \R_+ \times \R_{++}^2$, the equilibrium utilities of the validators in the corresponding stage game are
    \begin{align}
    \label{eq:eq-utility}
    u_i^*(x) := &  \frac{1}{N}\cdot \bigl((1-\theta) \cdot (\xi - \beta \cdot s^*(x))\cdot s^*(x) + \rho \cdot \price\cdot R\bigr) - \frac{\kappa}{2}\cdot \left(\frac{s^*(x)}{N}\right)^2
    \\
    = &  s^*(x) \cdot \frac{(1-\theta) \cdot \xi  + \rho \cdot \price\cdot R}{N} - (s^*(x))^2 \cdot \left( \frac{(1-\theta) \cdot \beta}{N} +  \frac{\kappa}{2N^2}\right) \ . 
\end{align} 
In particular, in any non-interior equilibrium, the utilities are equal to 
\begin{align}
    \label{eq:eq-utility-non-interior}
    u_i^*(x) = &  \frac{1}{N}\cdot \bigl(\rho \cdot \price\cdot R\bigr) - \frac{\kappa}{2}\cdot \left(\frac{\xi}{N \cdot \beta}\right)^2 \ . 
\end{align}

\end{corollary}

The second part of the corollary above implies that, for sufficiently large $R$, the  utility of each validator  is entirely reserve-driven. This is because the resulting equilibrium must be non-interior, and the corresponding market-clearing fee \eqref{eq:token-fee} is zero. The next corollary deals with the opposite case of small reserve (intuitively, when $R$ is not large enough to make the utility in \eqref{eq:eq-utility-non-interior} nonnegative).

\begin{corollary}\label{cor:interior} For sufficiently small $R$, the unique symmetric equilibrium must be interior, that is, $$s^*(R,\price,\xi) <\frac{\xi}{\beta}.$$ In particular, this holds true for any $R< R_{\min} := \left(\frac{\xi}{N \cdot \beta}\right)^2 \cdot \frac{\kappa N}{2 \rho \cdot \price} = \frac{\kappa \xi^2}{2 \rho \cdot \price \cdot N \cdot \beta^2}$. 
\end{corollary}

\begin{proof}
    By contradiction, if $s^*(R,\price,\xi) \geq \frac{\xi}{\beta}$, then $u_i^*(x) = \frac{1}{N}\left( 0 + \rho \cdot \price\cdot R\right) - \frac{\kappa}{2}\cdot \left(\frac{\xi}{N \cdot \beta}\right)^2 < 0$. But validator \(i\) can deviate to \(a_i=0\), yielding utility \(0\).
Hence the boundary profile cannot be a Nash equilibrium.
Therefore the equilibrium must be interior.
\end{proof}

Note that for any interior equilibria and for $Q:=N \cdot (1-\theta) \cdot \xi$, we have  
\begin{align}
    \label{eq:s-star-rewritten-interior}
s^*(R,\price,\xi)
= & 
 \frac{Q + \sqrt{Q^2 + 4\cdot (\kappa + (N+1) \cdot (1-\theta)\cdot \beta) \cdot (N-1) \cdot \rho  \cdot \price \cdot  R}}{2(\kappa + (N+1) \cdot (1-\theta)\cdot \beta)} \ .
\end{align}

\begin{example}[No reserve]\label{ex:noreserve}
Consider \(R=0\), so validator rewards are funded entirely by fees. By
Corollary~\ref{cor:interior}, the equilibrium is interior. Hence, for $x = (0,\price,\xi)$ we have 
\begin{align}\label{eq:equilibrium-no-reserve}
s^*(x)
\stackrel{\eqref{eq:s-star-rewritten-interior}}{=}
\frac{N\cdot (1-\theta)\cdot \xi}{
\kappa+(N+1)\cdot (1-\theta)\cdot \beta
} \stackrel{\eqref{eq:A-B}}{=} \frac{(1-\theta) \cdot \xi}{A_N} \ .
\end{align}
Next, we calculate the utilities of the validators. By Corollary~\ref{cor:utilities-equilibrium},
with $R=0$
\begin{align*}
u_i^*(x)
= & 
\frac{1}{N}\cdot (1-\theta)\cdot(\xi-\beta \cdot s^*(x))\cdot s^*(x)
-
\frac{\kappa}{2}\cdot \left(\frac{s^*(x)}{N}\right)^2 
\\
\stackrel{\eqref{eq:equilibrium-no-reserve}}{=} & 
\frac{1}{N} \cdot s^*(x) \cdot (A_N -(1-\theta)\cdot \beta)\cdot s^*(x)
-
\frac{\kappa}{2}\cdot \left(\frac{s^*(x)}{N}\right)^2
\\
= & 
\left(\frac{s^*(x)}{N}\right)^2 \cdot (N \cdot A_N -N\cdot (1-\theta)\cdot \beta)
-
\frac{\kappa}{2}\cdot \left(\frac{s^*(x)}{N}\right)^2
\\
\stackrel{\eqref{eq:A-B}}{=} & 
\left(\frac{s^*(x)}{N}\right)^2 \cdot (\kappa + (N+1)\cdot (1-\theta) \cdot \beta -N\cdot (1-\theta)\cdot \beta)
-
\frac{\kappa}{2}\cdot \left(\frac{s^*(x)}{N}\right)^2
\\
= & 
\left(\frac{s^*(x)}{N}\right)^2 \cdot (\frac{\kappa}{2} +  (1-\theta) \cdot \beta ) > 0 \ . 
\end{align*}
Thus, when \(R=0\), the equilibrium is interior and each validator obtains
strictly positive utility.
\end{example}

The next result shows that equilibrium security increases when validator rewards become easier to finance. Higher demand raises current fee revenue, and a higher token price makes reserve payouts more valuable in external terms. By contrast, higher operating costs, stronger fee compression, or diverting a larger fraction  of the fees  into the reserve reduce the validators' current incentives to supply security and blockspace.

\begin{proposition}[Monotonicity]
\label{prop:monotone}
The equilibrium security level $s^*(R,\price,\xi)$ from \Cref{thm:unique-sne} is weakly increasing in $R$, in $\price$, and in $\xi$, and weakly decreasing in $\kappa$, in $\beta$, and in $\theta$.
\end{proposition}

\subsection{Security Threshold}
In this section, we analyze the conditions under which the system sustains an equilibrium level of security (blockspace) above a threshold.
Let $\underline{s}>0$ denote the target security requirement.

\begin{definition}[Security-feasible state]
\label{def:security-feasible}
A state $(R,\price,\xi)$ is \emph{security feasible} for target $\underline{s}$ if
\[
s^*(R,\price,\xi)\ge \underline{s}.
\]
It is \emph{fee-only feasible} if the same inequality holds at $R=0$.
\end{definition}

The next theorem is the static core of the paper. It converts the equilibrium formula into a reserve threshold.

\begin{theorem}[Exact reserve threshold]
\label{thm:threshold}
For every target $\underline{s}>0$, define
\begin{equation}
\label{eq:threshold}
\SR(\price,\xi;\underline{s})
:=
\begin{cases}
\dfrac{\left[A_N\cdot \underline{s}^2-(1-\theta)\cdot \xi \cdot \underline{s}\right]_+}{B_N \cdot \price},
& \text{if } \underline{s}\le \xi/\beta,\\[1.2em]
\infty,
& \text{if } \underline{s}>\xi/\beta.
\end{cases}
\end{equation}
Then the following are equivalent:
\begin{enumerate}[label=(\roman*)]
\item $(R,\price,\xi)$ is security feasible for target $\underline{s}$.
\item $R\ge \SR(\price,\xi;\underline{s})$.
\end{enumerate}
Moreover, $(R,\price,\xi)$ is fee-only feasible if and only if $\SR(\price,\xi;\underline{s})=0$.
\end{theorem}

\Cref{thm:threshold} isolates the object that governance actually needs to monitor: not the reserve level alone, but the reserve level relative to a state-dependent threshold $\SR(\cdot)$. When demand is strong or the token price is high, the threshold falls because either current fees or reserve purchasing power are more favorable. When demand is too weak to support the target capacity, no reserve level can repair the shortfall. The threshold becomes infinite when the target security level $\underline{s}$ exceeds current demand capacity $\xi/\beta$ (recall that by \eqref{eq:symmetric-feasibility-bound} this ratio is the largest feasible $s$, i.e.,  the protocol cannot buy more secure throughput at any nonnegative fee).

\begin{remark}
The threshold $\SR(\cdot)$ is easiest to interpret geometrically. \Cref{fig:threshold-map} plots the exact threshold over the $(\price,\xi)$ plane for an illustrative parameterization. The figure makes visible the three economically distinct regions identified by \Cref{thm:threshold}: (i) an infeasible low-demand region, (ii)  a strictly reserve-dependent transition region, and (iii) a fee-only region where the threshold collapses to zero.
\end{remark}

\begin{figure}[tbp]
\centering
\includegraphics[width=\textwidth]{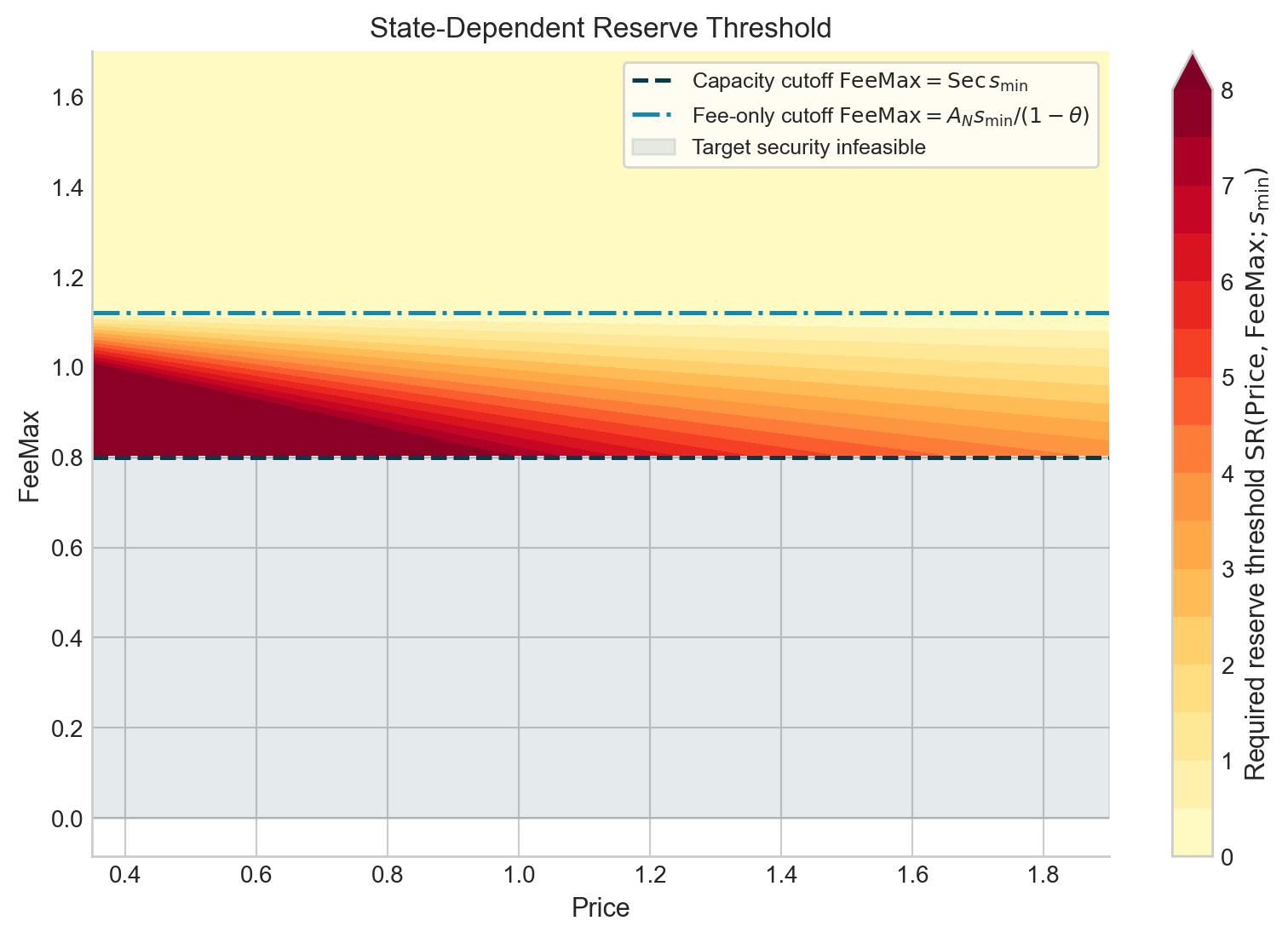}
\caption{Exact reserve threshold as a function of token price $\price$ and demand state $\xi$. For any point in the $(\price,\xi)$ plane, the color gives the minimal reserve stock required to sustain the target security level $\underline{s}$: darker regions require larger reserves. The gray lower region lies below the capacity cutoff $\xi=\beta\cdot \underline{s}$; there the target is infeasible at any finite reserve level because aggregate supply cannot be increased enough while keeping fees nonnegative. Above the fee-only cutoff $\xi=A_N\cdot \underline{s}/(1-\theta)$, the reserve threshold is zero because current fees alone can sustain the target. The figure uses the illustrative parameter values $N=5$, $\theta=0.25$, $\rho=0.04$, $\kappa=0.6$, $\beta=0.8$, and $\underline{s}=1$.}
\label{fig:threshold-map}
\end{figure}

\section{Dynamic Security Runway}\label{sec:dynamic}

The static threshold becomes dynamic once price and demand evolve over time and the reserve follows \Cref{eq:reserve-law}. This section formalizes the hand-off problem.

\subsection{Failure Time and Hand-Off Time}

Given the equilibrium of the policy from \Cref{thm:unique-sne}, the reserve evolution in \eqref{eq:reserve-law} becomes
\begin{equation}
\label{eq:equilibrium-reserve}
R_{t+1}
=
(1-\rho)R_t
+
\theta \frac{\big(\xi_t-\beta\cdot s^*(R_t,\price_t,\xi_t)\big)s^*(R_t,\price_t,\xi_t)}{\price_t}.
\end{equation}

\begin{definition}[Runway stopping times]
\label{def:stopping-times}
Fix a target security level $\underline{s}>0$. Define
\begin{align}
\tau_{\mathrm{fail}}
&:=
\inf\{t\ge 0: R_t < \SR(\price_t,\xi_t;\underline{s})\}, \label{eq:tau-fail}\\
\tau_{\mathrm{hand}}
&:=
\inf\{t\ge 0: \SR(\price_t,\xi_t;\underline{s})=0\}. \label{eq:tau-hand}
\end{align}
We say that the protocol achieves a \emph{successful hand-off} if
$
\tau_{\mathrm{hand}} < \tau_{\mathrm{fail}}.
$
\end{definition}

Next, we turn the hand-off problem into a hitting-time problem. Security lasts exactly as long as the reserve process remains inside the safe region, and the transition succeeds exactly when the protocol reaches the fee-only region before leaving that safe region.

\begin{proposition}[Exact runway criterion]
\label{prop:runway}
For every horizon $T\in \mathbb{N}$ and every sample path:
\begin{enumerate}[label=(\roman*)]
\item Security is maintained at all dates $t=0,\dots,T$ if and only if
\[
R_t \ge \SR(\price_t,\xi_t;\underline{s})
\qquad\text{for all } t=0,\dots,T.
\]
\item A successful hand-off occurs if and only if the process reaches the fee-only region before the failure region, that is,
\[
\tau_{\mathrm{hand}} < \tau_{\mathrm{fail}}.
\]
\end{enumerate}
\end{proposition}

\begin{proof}
Part (i) follows directly from \Cref{thm:threshold} applied state by state. Part (ii) is the definition of successful hand-off expressed in terms of the stopping times in Definition~\ref{def:stopping-times}.
\end{proof}

\subsection{A Conservative Stress-Test Bound}

The exact criterion is pathwise. For protocol design one often wants a finite-horizon sufficient condition that can be checked before the system is launched. The next result gives such a bound.

\begin{lemma}[Pure-decay lower bound]
\label{lem:pure-decay}
Along every sample path and for every $t\ge 0$,
\[
R_t \ge (1-\rho)^t \cdot R_0.
\]
\end{lemma}

\begin{proof}
The reserve recursion in \Cref{eq:equilibrium-reserve} has a nonnegative inflow term. Hence
$
R_{t+1}\ge (1-\rho)\cdot R_t.
$
Iterating yields the claim.
\end{proof}

\Cref{lem:pure-decay} is the worst-case reserve benchmark: even if future fee inflows are ignored entirely, the reserve cannot decay faster than pure payout at rate $\rho$. This simple lower bound is what makes conservative stress testing possible.

The next theorem provides a robust sufficient condition: if the reserve is large enough to survive a pessimistic lower envelope for price and demand, then the protocol is safe on every path inside that envelope. This form is useful for ex ante certification and governance stress tests.

\begin{theorem}[Finite-horizon stress-test guarantee]
\label{thm:stress-test}
Fix a horizon $T\in \mathbb{N}$ and deterministic lower envelopes
$(\underline{\price}_t)_{t=0}^T$ and
$(\underline{\xi}_t)_{t=0}^T.$
Define the event
$$E_T := \{\price_t \ge \underline{\price}_t,\ \xi_t \ge \underline{\xi}_t \ \ \  \text{ for all } t=0,\dots,T\}.$$
If
\begin{equation}
\label{eq:stress-bound}
(1-\rho)^t R_0 \ge \SR(\underline{\price}_t,\underline{\xi}_t;\underline{s})
\qquad \text{for all } t=0,\dots,T,
\end{equation}
then event $E_T$ implies that failure occurs after $T$, that is, 
$\tau_{\mathrm{fail}} > T.$
Consequently, for any $\alpha \in [0,1]$, if $\Prob(E_T)\ge 1-\alpha$, then
$\Prob(\tau_{\mathrm{fail}}>T)\ge 1-\alpha.$
\end{theorem}

\begin{remark}
\Cref{thm:stress-test} is intended for protocol governance. One may estimate lower confidence bands for $\price_t$ and $\xi_t$, plug them into \Cref{eq:stress-bound}, and obtain a conservative reserve requirement for a target horizon $T$. The bound is conservative because it ignores future reserve inflows from fees; it treats the reserve as if it only decayed. That makes it suitable for adverse-scenario certification.
\end{remark}

\subsection{Fee-Only Region and Long-Run Design}
We identify the economic moment at which the reserve becomes redundant, that is, current demand is  enough to finance the target security level. In this case, reserve policy no longer determines feasibility.

\begin{corollary}[Fee-only region]
\label{cor:fee-only}
Suppose $\underline{s}\le \xi/\beta$. Then fee-only feasibility at target $\underline{s}$ holds if and only if
\begin{equation}
\label{eq:fee-only-region}
(1-\theta)\cdot \xi \ge A_N \cdot  \underline{s}.
\end{equation}
\end{corollary}

\begin{proof}
By \Cref{thm:threshold}, fee-only feasibility is equivalent to $\SR(\price,\xi;\underline{s})=0$. Under $\underline{s}\le \xi/\beta$, this is equivalent to
\[
A_N\cdot \underline{s}^2-(1-\theta)\xi \cdot \underline{s}\le 0.
\]
Because $\underline{s}>0$, division by $\underline{s}$ yields \Cref{eq:fee-only-region}.
\end{proof}

\begin{remark}
\Cref{eq:fee-only-region} says that fee-only sustainability is governed by demand, \emph{not} by token price. This is because current transaction fees paid by  users are already measured in the external numeraire. Token-price risk matters only because the reserve is a stock of tokens (carried from the past).
\end{remark}

\section{Dynamic Validator Participation: A Markov Extension}\label{sec:Markov}

The baseline model in considered so far is  myopic: At state $X_t$, validators choose current participation to maximize current payoff only. Intuitively, this is what makes the reserve threshold explicit. In this section, we consider a fully strategic dynamic extension that endogenizes continuation values while preserving a precise equilibrium concept.

There are two technical issues. First, once validators become forward looking, a deviation at time $t$ changes the future reserve stock and therefore future incentives. Second, proving equilibrium existence in the original continuous-state model requires a substantial stochastic-game apparatus that would obscure the runway results. We therefore proceed in two steps. We first formulate a finite-horizon finite-state Markov game and prove existence of a Markov perfect equilibrium. We then return to the continuous model and derive the exact first-order condition that any differentiable pure Markov equilibrium must satisfy.

\subsection{Finite-Horizon Markov Game}

Fix a horizon $T\in \mathbb{N}$ and a discount factor $\delta \in (0,1)$.

\begin{assumption}[Finite-state Markov extension]
\label{ass:finite-extension}
There is a finite set of exogenous states
\[
\mathcal{Z}=\{z^1,\dots,z^M\}\subset \R_{++}^2,
\qquad
z^m=(\price^m,\xi^m),
\]
and a Markov transition matrix $P=(\fee_{mn})_{m,n=1}^M$ on $\mathcal{Z}$.
There is also a finite reserve grid
\[
\mathcal{R}=\{R^1,\dots,R^L\}\subset \R_+,
\qquad
0=R^1<\cdots<R^L=\bar R,
\]
with $\price^{\min}:=\min_{1\le m\le M} \price^m, \ \xi^{\max}:=\max_{1\le m\le M} \xi^m$ and
\begin{align*}
\bar R \ge \max\left\{R_0,\frac{\theta \cdot(\xi^{\max})^2}{4\beta\cdot \rho \cdot \price^{\min}}\right\}. 
\end{align*}
At state $x=(R,z)=(R,\price,\xi)\in \mathcal{R}\times\mathcal{Z}$, each validator chooses an action in the compact interval
\[
\mathcal{A}(x):=\left[0,\frac{\xi}{N\beta}\right].
\]
The next-period reserve is computed by first determining the intermediate value
\[
\widehat{R}'(x,a)
:=
(1-\rho)\cdot R + \theta \cdot\frac{(\xi-\beta\cdot s)\cdot s}{\price},
\qquad
s := \sum_{i=1}^N a_i,
\]
and subsequently projecting this value onto the discrete reserve grid via
\[
\Pi_{\mathcal{R}}(r) := \max\{R^\ell \in \mathcal{R} : R^\ell \le r\}.
\]
Accordingly, the next-period reserve is given by
\[
R'(x,a) := \Pi_{\mathcal{R}}(\widehat{R}'(x,a)).
\]
\end{assumption}

\begin{remark}
\Cref{ass:finite-extension} is a computational extension of the baseline model. The exogenous Markov chain is a finite approximation of the original price-demand process, and the reserve grid is the standard discretization used in dynamic policy computation. The bound on $\bar R$ guarantees that the reserve state is invariant: because $(\xi-\beta\cdot s)\cdot s\le \xi^2/(4\beta)$ and $\price\ge \price^{\min}$,
\[
\widehat{R}'(x,a)
\le
(1-\rho)\cdot \bar R + \frac{\theta \cdot (\xi^{\max})^2}{4\beta\cdot \price^{\min}}
\le
\bar R.
\]
\end{remark}

\begin{definition}[Markov strategy and Markov perfect equilibrium]
\label{def:mpe}
For each date $t\in\{0,\dots,T\}$ and each validator $i$, a (mixed) \emph{Markov strategy} is a mapping
\[
\sigma_t^i:\mathcal{R}\times\mathcal{Z}\to \Delta(\mathcal{A}(x)),
\]
where $\Delta(\mathcal{A}(x))$ denotes the set of Borel probability measures on $\mathcal{A}(x)$.

Given a strategy profile $\sigma=(\sigma_t^i)_{i,t}$, define continuation values recursively by
\[
V_{i,T+1}^\sigma(x):=0,
\]
and, for $t=T,T-1,\dots,0$,
 \begin{align}
V^\sigma_{i,t}(x)
=
\mathbb{E}_{a_1\sim \sigma_t^1(\cdot\mid x),\,\dots,\,a_N\sim \sigma_t^N(\cdot\mid x)}
\left[
u_i(a;x)
+
\delta \sum_{z'\in Z} P(z,z')\,V^\sigma_{i,t+1}\!\left(R'(x,a),z'\right)
\right].\label{eq:dynamic-value}
\end{align}

A strategy profile $\sigma$ is a \emph{Markov perfect equilibrium (in mixed strategies)} if for every date $t$, every state $x$, every validator $i$, and every alternative mixed action $\mu \in \Delta(\mathcal{A}(x))$,
\begin{align}
& \mathbb{E}_{a_1\sim \sigma_t^1(\cdot\mid x),\,\dots,\,a_N\sim \sigma_t^N(\cdot\mid x)}
\left[
u_i(a;x)
+
\delta \sum_{z'\in Z} P(z,z')\,V^\sigma_{i,t+1}\!\left(R'(x,a),z'\right)
\right]
\nonumber\\
&\qquad\qquad\qquad\ge
\mathbb{E}_{\substack{a_i\sim \mu\\ a_j\sim \sigma_t^j(\cdot\mid x),\, j\neq i}}
\left[
u_i(a;x)
+
\delta \sum_{z'\in Z} P(z,z')\,V^\sigma_{i,t+1}\!\left(R'(x,a),z'\right)
\right].
\label{eq:mpe-ineq}
\end{align}
\end{definition}

The next theorem shows that once the state space is discretized and the horizon is finite, forward-looking validator behavior can still be analyzed with a well-defined equilibrium concept.
\begin{theorem}[Existence of finite-horizon Markov perfect equilibrium]
\label{thm:mpe-existence}
Under \Cref{ass:finite-extension}, the finite-horizon dynamic validator game admits a Markov perfect equilibrium in mixed strategies.
\end{theorem}

The proof of the above theorem is a backward-induction argument: each date-$t$ problem becomes an ordinary continuation game once later continuation values are fixed.

The next result says that dynamic multiplicity does not arise if every continuation game is already pinned down locally. Once each state-date problem has a unique symmetric pure best-response fixed point, backward induction propagates that uniqueness through the full dynamic game.

\begin{proposition}[Purity and uniqueness under state-by-state uniqueness]
\label{prop:mpe-uniqueness}
Define the date-$t$ continuation-game payoff by
\begin{equation}
\label{eq:continuation-game-payoff}
g_{i,t}(x,a)
:=
u_i(a;x)
+
\delta \sum_{z' \in \mathcal{Z}}
P(z,z')V_{i,t+1}^\sigma\big((R'(x,a),z')\big).
\end{equation}
Suppose that, for every date $t\in\{0,\dots,T\}$ and every state $x\in \mathcal{R}\times\mathcal{Z}$, the continuation game with payoff \Cref{eq:continuation-game-payoff} admits a unique symmetric pure Nash equilibrium action $a_t^*(x)\in \mathcal{A}(x)$. Then the dynamic validator game admits a unique symmetric pure Markov perfect equilibrium, namely the profile that assigns action $a_t^*(x)$ to every validator at every state-date pair.
\end{proposition}

\begin{example}[Illustrating Proposition~\ref{prop:mpe-uniqueness}]
We construct a simple finite-horizon example in which the hypothesis of Proposition~\ref{prop:mpe-uniqueness} is satisfied.
Consider the finite-horizon Markov game with
\[
N=2,
\qquad
T<\infty,
\qquad
\delta\in(0,1),
\]
and let the exogenous state space \(Z\) and transition matrix \(P\) be arbitrary but finite, as in Assumption~\ref{ass:finite-extension}. Let the reserve grid \(\mathcal R\) also be arbitrary and finite. Now choose $\theta=0$.
Under this choice, the reserve transition becomes
\[
R'(x,a)=(1-\rho)R,
\]
so the next-period reserve no longer depends on the current action profile \(a\). Hence, for every date \(t\) and state \(x=(R,z)\), the continuation-game payoff
\[
g_{i,t}(x,a)
=
u_i(a;x)
+
\delta \sum_{z'\in Z} P(z,z')\,V^\sigma_{i,t+1}(R'(x,a),z')
\]
can be written as
\[
g_{i,t}(x,a)
=
u_i(a;x)+C_t(x),
\]
where \(C_t(x)\) is a constant with respect to the action profile \(a\). Therefore, the continuation game at any state-date pair \((t,x)\) has exactly the same best responses as the baseline stage game.
By Theorem~\ref{thm:unique-sne}, for every state \(x=(R,\price,\xi)\), the baseline stage game admits a unique symmetric Nash equilibrium, with aggregate security
\[
s^*(x)
=
\min\!\left\{
\frac{(1-\theta)\xi+\sqrt{(1-\theta)^2 \xi^2+4A_N\cdot B_N\cdot \price\cdot R}}{2A_N},
\frac{\xi}{\beta}
\right\},
\]
and corresponding individual action
$a^*(x)=\frac{s^*(x)}{N}.$
Since here \(\theta=0\), each continuation game therefore admits the same unique symmetric pure Nash equilibrium action $a_t^*(x)=a^*(x)$ for every $t\in\{0,\dots,T\},\ x\in \mathcal R\times Z$.
Thus, the hypothesis of Proposition~\ref{prop:mpe-uniqueness} is satisfied, and it follows that the dynamic validator game admits a unique symmetric pure Markov perfect equilibrium.
\end{example}

\subsection{Continuous-State First-Order Condition}

The finite-state extension above gives an existence theorem. To understand economically how forward-looking incentives modify the baseline equilibrium, it is useful to return to the original continuous model and derive the equilibrium condition that must hold whenever a pure differentiable Markov equilibrium exists.

The next result shows exactly how forward-looking incentives perturb the myopic equilibrium.
\begin{proposition}[Dynamic symmetric first-order condition]
\label{prop:dynamic-foc}
Fix a horizon $T\ge 1$ and suppose that, in the original continuous model, there exists a pure symmetric Markov perfect equilibrium with continuation value functions
\[
V_{t+1}(R,\price,\xi),
\qquad
t=0,\dots,T-1,
\]
that are continuously differentiable in $R$. Fix a date $t<T$ and a state $x=(R,\price,\xi)$. If the equilibrium at that state is interior, with aggregate security $s_t^{\mathrm{M}}(x)\in (0,\xi/\beta)$, then it satisfies
\begin{equation}
\label{eq:dynamic-foc}
(1-\theta)\xi
-
A_N \cdot s
+
\frac{B_N \cdot \price \cdot R}{s}
+
\frac{\delta \cdot\theta}{\price}
\left(\xi-\frac{N+1}{N}\beta\cdot s\right)
M_{t+1}(x;s)
=
0,
\end{equation}
where $s=s_t^{\mathrm{M}}(x)$ and
\begin{align}
& M_{t+1}(x;s)
:= \nonumber \\
& 
\label{eq:M-def}
\E\!\left[
\partial_R V_{t+1}\!\left(
(1-\rho)R + \theta \frac{(\xi-\beta\cdot s)s}{\price},
\price_{t+1},
\xi_{t+1}
\right)
\,\middle|\,
\substack{
\price_t=\price,\\
\xi_t=\xi
}
\right].
\end{align}
\end{proposition}

\begin{remark}
The first three terms in \Cref{eq:dynamic-foc} are exactly the myopic first-order condition from the baseline model. The final term is new. It is the marginal value of the effect of current participation on the next reserve stock. If $M_{t+1}(R,\price,\xi;s)\ge 0$, then forward-looking validators value reserve accumulation. The sign of the whole term is then governed by
\[
\xi-\frac{N+1}{N}\beta\cdot s,
\]
which is the marginal effect of one validator's additional participation on next period's reserve inflow. When aggregate participation is below the point at which extra participation sharply compresses fees, dynamic incentives push equilibrium security upward relative to the myopic benchmark. When the system is already close to maximum congestion, the intertemporal effect becomes weaker and can eventually reverse.
\end{remark}

The next proposition shows that the runway logic survives the Markov extension.
\begin{proposition}[Runway under a pure Markov equilibrium]
\label{prop:markov-runway}
Suppose the pure symmetric Markov perfect equilibrium from \Cref{prop:dynamic-foc} exists and denote its aggregate security policy at date $t$ by
\[
s_t^{\mathrm{M}}(R,\price,\xi).
\]
Define the date-$t$ security set
\[
\mathcal{S}_t(\underline{s})
:=
\{(R,\price,\xi)\in \R_+\times\R_{++}^2: s_t^{\mathrm{M}}(R,\price,\xi)\ge \underline{s}\}.
\]
Then, along any equilibrium sample path, security is maintained through date $T$ if and only if
\[
(R_t,\price_t,\xi_t)\in \mathcal{S}_t(\underline{s})
\qquad \text{for all } t=0,\dots,T.
\]
If, in addition, $s_t^{\mathrm{M}}(R,\price,\xi)$ is weakly increasing in $R$ for each fixed $(\price,\xi)$, then the dynamic reserve threshold
\[
\SR_t^{\mathrm{M}}(\price,\xi;\underline{s})
:=
\inf\{R\ge 0: s_t^{\mathrm{M}}(R,\price,\xi)\ge \underline{s}\}
\]
is well defined, and the runway criterion can be written as
\[
R_t \ge \SR_t^{\mathrm{M}}(\price_t,\xi_t;\underline{s})
\qquad \text{for all } t=0,\dots,T.
\]
\end{proposition}

 The threshold $\SR_t^{\mathrm{M}}(\cdot)$ need no longer have a closed form, but once equilibrium participation is monotone in reserves, the safe region is still summarized by a reserve cutoff at each state and date.

\section{Failure Probabilities and Expected Hand-Off Times}\label{sec:Brownian}

The main results identify the safe region state by state and path by path without committing to a particular stochastic law for price and demand. For quantitative policy analysis, however, one often wants explicit probabilistic outputs such as failure probabilities over a fixed horizon or expected time to fee-only operation. This section derives such metrics under a parametric specialization.

\begin{assumption}[Lognormal state dynamics]
\label{ass:parametric}
There exist constants $\mu_{\price},\mu_{\xi} \in \R$ and $\sigma_{\price},\sigma_{\xi}>0$, a standard Brownian motion $(W_t)_{t\ge 0}$, and an i.i.d. sequence
\[
(\Delta W_t,\eta_t)_{t\ge 1}
\]
of centered bivariate normal random vectors with unit marginal variances, where $\Delta W_t:=W_t-W_{t-1}$ for $t\ge 1$, such that, for every $t\ge 0$,
\begin{align}
\log \price_t &= \log \price_0 + \left(\mu_{\price}-\frac{\sigma_{\price}^2}{2}\right)t + \sigma_{\price} \cdot  W_t, \label{eq:parametric-q}\\
\log \xi_t &= \log \xi_0 + \mu_{\xi} t + \sigma_{\xi} \sum_{k=1}^t \eta_k. \label{eq:parametric-xi}
\end{align}
We write $\Phi$ for the standard normal cumulative distribution function (cdf).
\end{assumption}

\begin{remark}[geometric Brownian motion]
\Cref{ass:parametric} is imposed only in this section. It strengthens \Cref{ass:state} in order to convert the law-free runway criteria into explicit probabilistic metrics. The price process is a geometric Brownian motion observed at integer dates, while demand remains in a discrete-time lognormal specification. Correlation between token-price shocks and demand shocks is allowed through the joint law of $(\Delta W_t,\eta_t)$; the results below use only the normal marginal distributions and a union-bound argument.
\end{remark}

Next, we convert the law-free stress-test criterion from \Cref{thm:stress-test} into an explicit statement.

\begin{theorem}[Finite-horizon failure probability under lognormal dynamics]
\label{thm:parametric-failure}
Fix a horizon $T\in \mathbb{N}$ and confidence parameters $z_{\price},z_{\xi}\ge 0$. Define deterministic lower envelopes by
\begin{align}
\underline{\price}_t(z_{\price}) &:= \price_0\cdot \exp\! \left(\left(\mu_{\price}-\frac{\sigma_{\price}^2}{2}\right)t-\sigma_{\price}\cdot \sqrt{t}\cdot z_{\price}\right), \label{eq:parametric-lower-q}\\
\underline{\xi}_t(z_{\xi}) &:= \xi_0\cdot \exp\!\left(\mu_{\xi} \cdot t-\sigma_{\xi}\cdot \sqrt{t}\cdot z_{\xi}\right), \label{eq:parametric-lower-xi}
\end{align}
for $t=0,\dots,T$. If
\begin{equation}
\label{eq:parametric-stress}
(1-\rho)^tR_0 \ge \SR(\underline{\price}_t(z_{\price}),\underline{\xi}_t(z_{\xi});\underline{s})
\qquad \text{for all } t=0,\dots,T,
\end{equation}
then
\begin{equation}
\label{eq:parametric-failure-bound}
\Prob(\tau_{\mathrm{fail}}>T)
\ge
1-T\cdot \big(\Phi(-z_{\price})+\Phi(-z_{\xi})\big).
\end{equation}
In particular, if
\[
z_{\price}=\Phi^{-1}\!\left(1-\frac{\alpha_{\price}}{T}\right),
\qquad
z_{\xi}=\Phi^{-1}\!\left(1-\frac{\alpha_{\xi}}{T}\right)
\]
for some $\alpha_{\price},\alpha_{\xi}\in(0,1)$, then
\[
\Prob(\tau_{\mathrm{fail}}>T)\ge 1-\alpha_{\price}-\alpha_{\xi}.
\]
\end{theorem}

\begin{proposition}[Hand-off probabilities and expected hand-off time]
\label{prop:parametric-handoff}
Suppose \Cref{ass:parametric} holds and define the fee-only demand cutoff
\begin{equation}
\label{eq:fee-only-cutoff}
\xi^{\mathrm{FO}}:=\frac{A_N\cdot \underline{s}}{1-\theta}.
\end{equation}
Let
\[
y_0:=\log \xi_0,
\qquad
y^{\mathrm{FO}}:=\log \xi^{\mathrm{FO}}.
\]
Assume $\mu_{\xi}>0$ and $\xi_0<\xi^{\mathrm{FO}}$. Then:
\begin{enumerate}[label=(\roman*)]
\item The hand-off time satisfies
\[
\tau_{\mathrm{hand}}
=
\inf\{t\ge 0:\xi_t\ge \xi^{\mathrm{FO}}\}
<
\infty
\qquad \text{almost surely.}
\]
\item For every horizon $T\in\mathbb{N}$,
\begin{equation}
\label{eq:parametric-handoff-prob}
\Prob(\tau_{\mathrm{hand}}\le T)
\ge
1-\Phi\!\left(\frac{y^{\mathrm{FO}}-y_0-\mu_{\xi} T}{\sigma_{\xi}\sqrt{T}}\right).
\end{equation}
\item The expected hand-off time is finite and satisfies
\begin{equation}
\label{eq:parametric-handoff-lower}
\frac{y^{\mathrm{FO}}-y_0}{\mu_{\xi}}
\le
\E[\tau_{\mathrm{hand}}]
\le
1+\sum_{t=1}^{\infty}
\Phi\!\left(\frac{y^{\mathrm{FO}}-y_0-\mu_{\xi} t}{\sigma_{\xi}\sqrt{t}}\right)
<
\infty.
\end{equation}
\item If, in addition, the condition of \Cref{thm:parametric-failure} holds for the same horizon $T$, then
\begin{equation}
\label{eq:parametric-success-prob}
\Prob(\tau_{\mathrm{hand}}<\tau_{\mathrm{fail}})
\ge
1-T\big(\Phi(-z_{\price})+\Phi(-z_{\xi})\big)
-\Phi\!\left(\frac{y^{\mathrm{FO}}-y_0-\mu_{\xi} T}{\sigma_{\xi}\sqrt{T}}\right).
\end{equation}
\end{enumerate}
\end{proposition}

\Cref{prop:parametric-handoff} separates two distinct quantitative questions. The first is how quickly demand alone is likely to carry the protocol into the fee-only region. The second is whether that arrival is likely to happen before reserve failure. Positive demand drift makes eventual hand-off almost sure in this parametric environment, while the finite-horizon bounds show how drift and volatility translate into operational success probabilities.

The analytical results above are complemented by two numerical illustrations. \Cref{fig:runway-scenarios} plots three deterministic state paths built from the parametric law: one path fails before hand-off, one reaches fee-only sustainability gradually, and one hands off quickly. \Cref{fig:probabilistic-metrics} then shows how the quantitative runway picture changes with demand drift in the parametric environment. The left panel reports Monte Carlo estimates of survival, hand-off, and successful-transition probabilities at a fixed horizon, while the right panel plots the analytical lower and upper bounds on expected hand-off time from \Cref{prop:parametric-handoff}.

\begin{figure}[tbp]
\centering
\includegraphics[width=\textwidth]{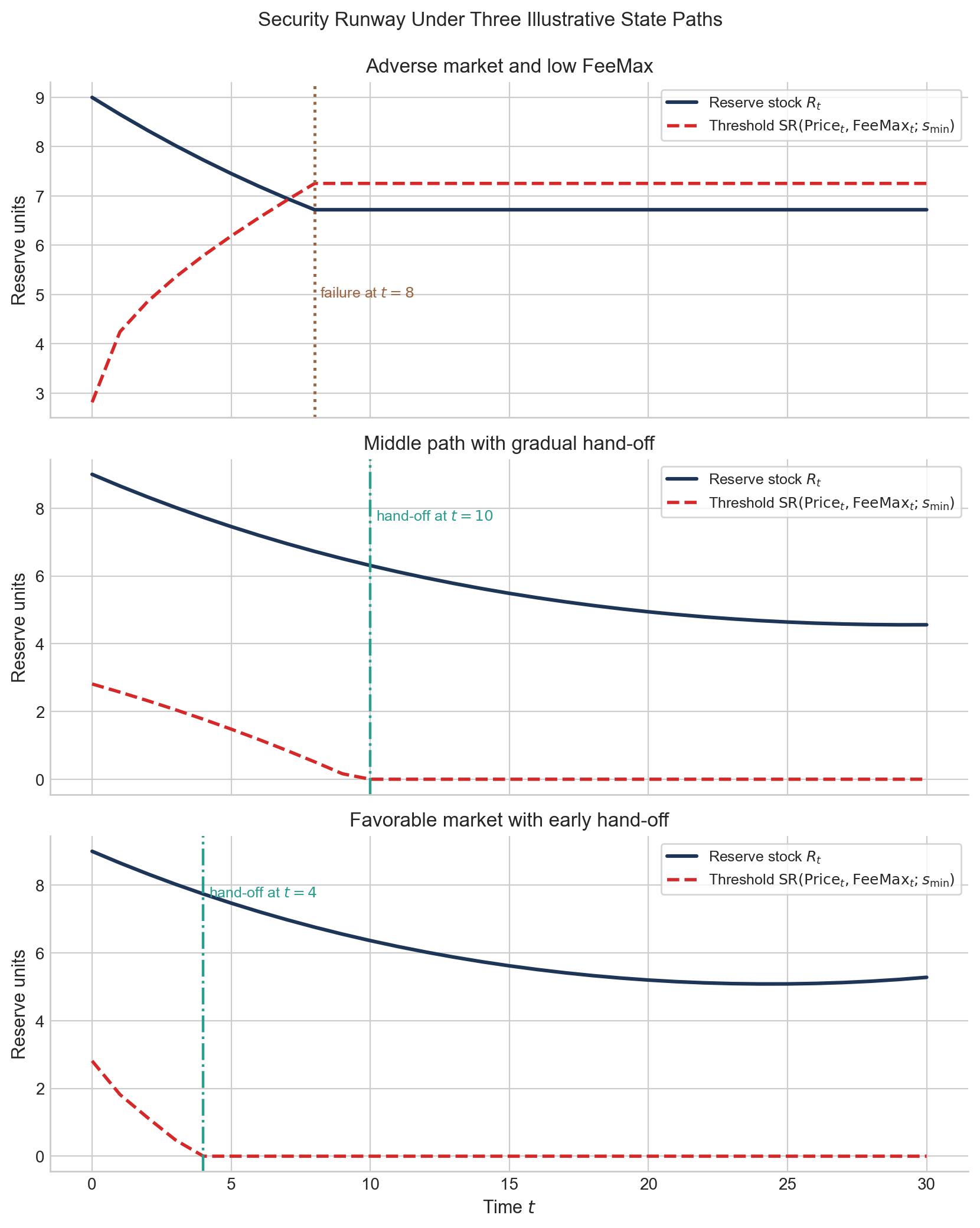}
\caption{Illustrative runway paths under three deterministic state trajectories generated from the parametric law. In each panel, the solid blue curve is the reserve stock $R_t$ and the dashed red curve is the state-dependent threshold $\SR(\price_t,\xi_t;\underline{s})$. The vertical marker identifies the first economically relevant event along the path: hand-off when the threshold reaches zero, or failure when the reserve falls below the threshold. The adverse path crosses the failure boundary before reaching the fee-only region; the middle path reaches fee-only sustainability only after a prolonged transition; and the favorable path hands off early. The figure uses the same protocol parameters as \Cref{fig:threshold-map}, with initial reserve $R_0=9$.}
\label{fig:runway-scenarios}
\end{figure}

\begin{figure}[tbp]
\centering
\includegraphics[width=\textwidth]{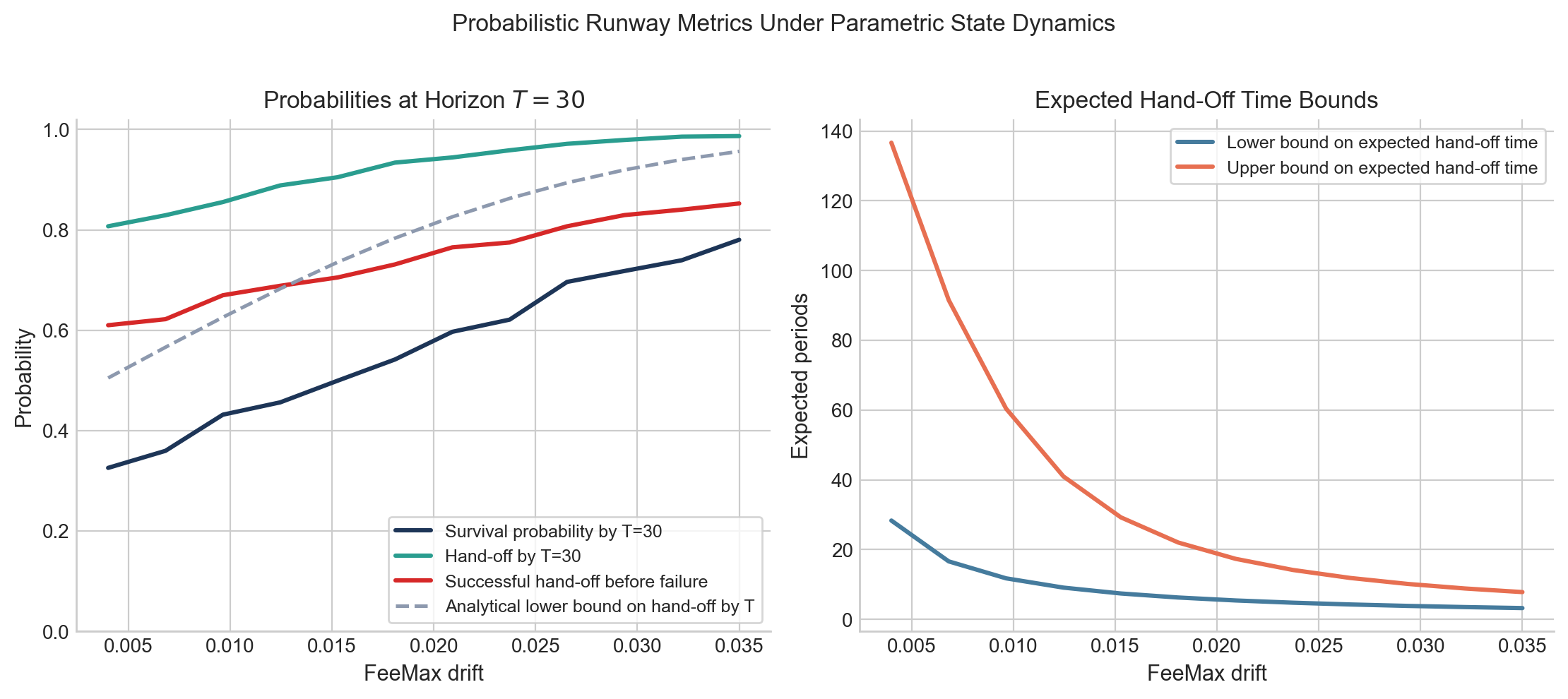}
\caption{Probabilistic runway metrics under the parametric state dynamics. The horizontal axis in both panels is the demand drift $\mu_{\xi}$. In the left panel, the solid curves are Monte Carlo estimates of survival probability by horizon $T=30$, hand-off probability by horizon $T=30$, and the probability of successful hand-off before failure. The dashed curve is the analytical lower bound on hand-off by time $T$ from \Cref{prop:parametric-handoff}. In the right panel, the two curves are the analytical lower and upper bounds on expected hand-off time from \Cref{prop:parametric-handoff}. The figure uses the same protocol parameters as \Cref{fig:threshold-map}, with initial reserve $R_0=9$, $\price_0=1$, $\xi_0=1$, price drift $\mu_{\price}=0.005$, price volatility $\sigma_{\price}=0.16$, and demand volatility $\sigma_{\xi}=0.10$. Higher demand drift increases survival and hand-off probabilities and narrows the expected time to hand-off.}
\label{fig:probabilistic-metrics}
\end{figure}

\section{Conclusion and Future Work}

Proof-of-stake systems with finite reserves face a reserve hand-off problem.
Early validator rewards may rely on reserve subsidies, while long-run security
must eventually be supported by transaction fees. This paper formalized this
transition through a stochastic model with token-price risk, demand risk, and
strategic validator participation.
We solved the symmetric validator participation game and derived a closed-form
equilibrium. Using this equilibrium, we obtained an explicit state-dependent
reserve threshold. For each token price and demand state, the threshold gives
the minimal reserve stock necessary and sufficient to sustain a target security
level. This threshold separates states in which security is infeasible, states
in which security depends on the reserve, and states in which fees alone are
sufficient.
We then used the threshold to study the dynamic hand-off problem. Security is
maintained along a path exactly while the reserve remains above the
state-dependent threshold. Failure occurs when the reserve first falls below
this threshold. A successful hand-off occurs when the system reaches the
fee-only region before failure. This converts reserve adequacy into a
state-dependent hitting-time problem.
We also studied forward-looking validator incentives. In the Markov extension,
current participation affects both current rewards and the next reserve stock,
and therefore changes future reward opportunities. We proved finite-horizon
Markov perfect equilibrium existence and derived the dynamic first-order
condition that captures this continuation-value effect. Finally, under
lognormal price-demand dynamics, we translated the runway analysis into
finite-horizon failure-probability bounds, hand-off probabilities, and bounds on
expected hand-off time.

The main conclusion is that reserves are not neutral accounting devices. Their
security value depends on the current state of the system. A protocol with a
large nominal reserve can still be close to failure after an adverse price or
demand shock. Conversely, a protocol with a smaller reserve can be safe once
demand is strong enough to support the target security level through fees alone.
Reserve policy should therefore be evaluated through state-dependent runway
analysis rather than through nominal depletion dates or steady-state reward
ratios alone.

The model made several simplifying assumptions. Token price and demand were
treated as exogenous, so the analysis did not model token valuation or user
adoption. Validators were symmetric in the baseline model, so the threshold
should be interpreted as a benchmark rather than as a full description of
heterogeneous operators. The finite-horizon stress test was deliberately
conservative because it lower-bounded the reserve path by ignoring future fee
inflows. These assumptions made it possible to isolate the reserve hand-off
mechanism and obtain explicit threshold conditions.

A natural next step is adaptive reserve policy. In this paper, the policy
parameters $(\theta,\rho)$ were fixed. This provided a clean benchmark for
analyzing the security-runway mechanism. In principle, a protocol may want
payout and fee-retention rules to respond to the current reserve, token price,
and demand state. For example, the protocol could reduce reserve payouts after
adverse price realizations, increase fee retention when the reserve is low, or
change its policy once demand approaches the fee-only region.

Such an extension would turn the analysis from policy evaluation into policy
design. Governance would choose a state-dependent rule of the form
\[
(\theta_t,\rho_t)=\phi(R_t,\price_t,\xi_t)
\]
to extend runway, reduce failure risk, or reach fee-only sustainability subject
to incentive and credibility constraints. Developing such a theory would require
combining the equilibrium framework in this paper with tools from dynamic
mechanism design or stochastic control. We view this as an important direction
for future work.
\bibliography{references}

\newpage
\appendix

\section{Probabilistic Setup}

This appendix explains the standard measurability conventions used implicitly in the main text.
Let
\[
\mathcal{G}_t:=\sigma\big((\price_0,\xi_0),\dots,(\price_t,\xi_t)\big),
\qquad t\ge 0.
\]
By \Cref{ass:state}, the exogenous state process $(\price_t,\xi_t)_{t\ge 0}$ is adapted to $(\mathcal{G}_t)_{t\ge 0}$.

The state variable
\[
X_t=(R_t,\price_t,\xi_t)
\]
is therefore assembled from two sources. The coordinates $(\price_t,\xi_t)$ come from the exogenous shock process specified in \Cref{ass:state}. The coordinate $R_t$ is endogenous: it starts from the initial condition $R_0$ and is generated recursively from past states and past equilibrium participation through the reserve law. Thus the state process is induced jointly by the primitive shock process and the reserve recursion under the equilibrium play.

Now define the reserve process recursively by
\[
R_{t+1}
=
(1-\rho)R_t+\theta \frac{(\xi_t-\beta \cdot s_t)s_t}{\price_t},
\]
where $R_0$ is $\mathcal{F}_0$-measurable and $s_t$ is the equilibrium security level generated by the model at date $t$. Since $s_t$ is a measurable function of the current state $(R_t,\price_t,\xi_t)$, it follows inductively that $(R_t)_{t\ge 0}$ is adapted to the filtration generated by the state process.

Accordingly, the full state process
\[
X_t=(R_t,\price_t,\xi_t)
\]
is adapted. The stopping times introduced in \Cref{def:stopping-times},
\begin{align*}
\tau_{\mathrm{fail}}
&=
\inf\{t\ge 0:R_t<\SR(\price_t,\xi_t;\underline{s})\},\\
\tau_{\mathrm{hand}}
&=
\inf\{t\ge 0:\SR(\price_t,\xi_t;\underline{s})=0\},   
\end{align*}
are therefore stopping times with respect to the filtration generated by $(X_t)_{t\ge 0}$, because the events defining them depend only on the current state.

\section{Postponed Proofs}

\subsection{Proof of Theorem~\ref{thm:unique-sne}}

\begin{proof}
Fix a state \(x=(R,\price,\xi)\) and suppress the state argument \(x\) in the
notation of utility \eqref{eq:utility}.

For validator \(i\), let
\[
A_{-i}:=\sum_{j\ne i}a_j .
\]
If \(A_{-i}>0\), then for every feasible \(a_i\),
\[
u_i(a_i,a_{-i})
=
(1-\theta)\cdot \xi \cdot a_i
-
(1-\theta)\cdot \beta \cdot a_i(a_i+A_{-i})
+
\frac{\rho \cdot  \price \cdot R\,a_i}{a_i+A_{-i}}
-
\frac{\kappa}{2}a_i^2 .
\]
Differentiating with respect to \(a_i\) gives
\begin{align}
\frac{\partial u_i}{\partial a_i}
&=
(1-\theta)\cdot \xi
-
(1-\theta)\cdot \beta \cdot (2a_i+A_{-i})
+
\rho \cdot  \price \cdot R\cdot \frac{A_{-i}}{(a_i+A_{-i})^2}
-
\kappa a_i, \label{eq:br-first-deriv}\\
\frac{\partial^2 u_i}{\partial a_i^2}
&=
-2(1-\theta)\cdot \beta
-
2\rho  \cdot \price \cdot R\cdot \frac{A_{-i}}{(a_i+A_{-i})^3}
-
\kappa
<0. \label{eq:br-second-deriv}
\end{align}
Hence \(u_i(\cdot,a_{-i})\) is strictly concave whenever \(A_{-i}>0\).

If \(A_{-i}=0\), then for \(a_i>0\),
\[
u_i(a_i,0_{-i})
=
(1-\theta)\cdot (\xi-\beta \cdot a_i)\cdot a_i
+
\rho \cdot  \price \cdot R
-
\frac{\kappa}{2}a_i^2,
\]
whose second derivative is \(-2(1-\theta)\cdot \beta-\kappa<0\).

The all-zero profile is not a Nash equilibrium. Indeed, if all other
validators choose zero, validator \(i\) can choose a sufficiently small
\(a_i>0\). Then
\[
u_i(a_i,0_{-i})
=
(1-\theta)\cdot (\xi-\beta \cdot a_i)\cdot a_i
+
\rho \cdot  \price \cdot R
-
\frac{\kappa}{2}a_i^2
>0
\]
for all sufficiently small \(a_i>0\), whereas
\(u_i(0,\ldots,0)=0\). Thus every symmetric equilibrium has strictly positive
aggregate security \(s>0\).

Existence of a symmetric Nash equilibrium follows from the compact and convex
common action set \([0,\xi/(N\beta)]\), symmetry, and the existence of a
maximizer of each validator's payoff on this compact set. We now characterize
all symmetric equilibria and show uniqueness.

Let \(a_i=a\) for all \(i\), and write \(s=Na\). Since every symmetric
equilibrium has \(s>0\), the first-order condition for an interior symmetric
equilibrium is well defined. Substituting \(a=s/N\) and
\(A_{-i}=(N-1)s/N\) into \eqref{eq:br-first-deriv} gives
\[
(1-\theta)\cdot\xi
-
\frac{N+1}{N}\cdot(1-\theta)\cdot\beta \cdot s
-
\frac{\kappa}{N}\cdot s
+
\frac{N-1}{N}\cdot\frac{\rho\cdot  \price \cdot R}{s}
=0.
\]
Equivalently, using \Cref{eq:A-B},
\[
(1-\theta)\cdot \xi
-
A_N \cdot  s
+
\frac{B_N \cdot \price \cdot R}{s}
=0.
\]
Multiplying by \(s>0\) gives
\begin{equation}
\label{eq:quadratic}
A_N\cdot s^2-(1-\theta)\cdot \xi \cdot s-B_N \cdot \price \cdot R=0.
\end{equation}
If \(B_N \cdot \price \cdot R>0\), this equation has exactly one positive solution. If
\(B_N \price \cdot R=0\), its roots are \(0\) and \((1-\theta)\xi/A_N\); the root \(0\)
is extraneous because the first-order condition was derived only for \(s>0\).
Hence the unique positive interior candidate is
\[
\widehat{s}(R,\price,\xi)
=
\frac{(1-\theta)\cdot \xi+
\sqrt{(1-\theta)^2\cdot \xi^2+4A_N\cdot B_N \cdot \price \cdot R}}
{2A_N}.
\]

Define
\[
F(s):=(1-\theta)\cdot \xi-A_N \cdot s+\frac{B_N \cdot \price \cdot R}{s},
\qquad s>0.
\]
Then
\[
F'(s)=-A_N-\frac{B_N \cdot \price \cdot R}{s^2}<0,
\]
so the symmetric first-order expression is strictly decreasing in \(s\).
Therefore there is at most one interior symmetric equilibrium.

If \(\widehat{s}(R,\price,\xi)\le \xi/\beta\), the capacity constraint does not
bind, and strict concavity implies that the unique symmetric equilibrium has
aggregate security
\[
s^*=\widehat{s}(R,\price,\xi).
\]

If \(\widehat{s}(R,\price,\xi)>\xi/\beta\), then
\[
F(\xi/\beta)>0,
\]
because \(F\) is strictly decreasing and its unique positive zero is
\(\widehat{s}(R,\price,\xi)\). Thus, at the symmetric boundary profile
\(a_i=\xi/(N\beta)\), each validator's payoff is still increasing in its own
action at the upper end of the feasible interval. By strict concavity, the
unique best response is therefore the boundary action itself. Hence the unique
symmetric equilibrium has aggregate security
\[
s^*=\xi/\beta .
\]

Combining the two cases,
\[
s^*(R,\price,\xi)
=
\min\left\{
\frac{(1-\theta)\xi+
\sqrt{(1-\theta)^2\xi^2+4A_NB_N \price \cdot R}}
{2A_N},
\frac{\xi}{\beta}
\right\}.
\]
 The formula for the corresponding equilibrium  fee $f^*(x)$ follows directly from \eqref{eq:token-fee}. 
\end{proof}

\subsection{Proof of Proposition~\ref{prop:monotone}}
\begin{proof}
On the interior branch of \Cref{eq:s-star},
\[
\widehat{s}(R,\price,\xi)
=
\frac{(1-\theta)\xi + \sqrt{(1-\theta)^2\xi^2 + 4A_N B_N  \price \cdot R}}{2A_N}.
\]
The square-root term is increasing in $ \price \cdot R$ and in $\xi$, so $\widehat{s}$ is increasing in $R$, $q$, and $\xi$. It is decreasing in $A_N$, hence decreasing in $\kappa$, $\beta$, and $\theta$, and it is also decreasing in $\theta$ directly through the term $(1-\theta)\xi$. The full equilibrium $s^*$ is the minimum of $\widehat{s}$ and the capacity bound $\xi/\beta$, which preserves all weak monotonicity statements.
\end{proof}

\subsection{Proof of Theorem~\ref{thm:threshold}}
\begin{proof}
Suppose first that $\underline{s}>\xi/\beta$. Since \Cref{thm:unique-sne} implies $s^*(R,\price,\xi)\le \xi/\beta$ for every $R$, security feasibility is impossible, which matches $\SR(\price,\xi;\underline{s})=\infty$.

Now suppose $\underline{s}\le \xi/\beta$. By \Cref{thm:unique-sne}, security feasibility is equivalent to
\[
\widehat{s}(R,\price,\xi)\ge \underline{s},
\]
where $\widehat{s}$ denotes the interior root from the proof of \Cref{thm:unique-sne}. Using the explicit formula,
\[
\frac{(1-\theta)\xi + \sqrt{(1-\theta)^2\xi^2 + 4A_N B_N  \price \cdot R}}{2A_N}
\ge
\underline{s}.
\]
Rearranging yields
\[
\sqrt{(1-\theta)^2\xi^2 + 4A_N B_N  \price \cdot R}
\ge
2A_N\underline{s}-(1-\theta)\xi.
\]
If $2A_N\underline{s}-(1-\theta)\xi \le 0$, then the inequality holds automatically, and security is feasible even at $R=0$. This is exactly the case $\left[A_N\underline{s}^2-(1-\theta)\xi\underline{s}\right]_+=0$.

If $2A_N\underline{s}-(1-\theta)\xi>0$, both sides are nonnegative, so squaring is valid and gives
\[
(1-\theta)^2\xi^2 + 4A_N B_N  \price \cdot R
\ge
\left(2A_N\underline{s}-(1-\theta)\xi\right)^2.
\]
After cancellation,
\[
4A_N B_N  \price \cdot R
\ge
4A_N\left(A_N\underline{s}^2-(1-\theta)\xi\underline{s}\right).
\]
Since $A_N>0$, this is equivalent to
\[
R \ge \frac{A_N\underline{s}^2-(1-\theta)\xi\underline{s}}{B_N q}.
\]
Combining the two cases gives \Cref{eq:threshold}.

Finally, fee-only feasibility means security feasibility at $R=0$, which is equivalent to the threshold being zero.
\end{proof}

\subsection{Proof of Theorem~\ref{thm:stress-test}}
\begin{proof}
On $E_T$, we have $ \price_t\ge \underline{\price}_t$ and $\xi_t\ge \underline{\xi}_t$ for each $t\le T$. By \Cref{eq:threshold}, the threshold $\SR(\price,\xi;\underline{s})$ is weakly decreasing in both $q$ and $\xi$, so
\[
\SR( \price_t,\xi_t;\underline{s})
\le
\SR(\underline{\price}_t,\underline{\xi}_t;\underline{s})
\qquad \text{for all } t\le T.
\]
By \Cref{lem:pure-decay},
\[
R_t \ge (1-\rho)^t R_0.
\]
Combining this with \Cref{eq:stress-bound} gives
\[
R_t \ge \SR(\price_t,\xi_t;\underline{s})
\qquad \text{for all } t\le T
\]
on $E_T$. Therefore $\tau_{\mathrm{fail}}>T$ on $E_T$ by \Cref{prop:runway}. The probability statement follows immediately.
\end{proof}

\subsection{Proof of Theorem~\ref{thm:mpe-existence}}
\begin{proof}
We proceed by backward induction.

\emph{Step 1: terminal date.}
Fix a state $x=(R,z)\in \mathcal{R}\times\mathcal{Z}$. At date $T$, the continuation term is zero, so the date-$T$ continuation game is the one-shot normal-form game with action sets $\mathcal{A}(x)$ and payoffs $u_i(\cdot;x)$. The action sets are nonempty compact intervals and the payoff functions are continuous. Therefore the mixed-strategy equilibrium existence theorem for continuous games yields a mixed Nash equilibrium at date $T$ for every state $x$. Select one such equilibrium and denote it by $\sigma_T(\cdot\mid x)$. This determines $V_{i,T}^\sigma(x)$ through \Cref{eq:dynamic-value}.

\emph{Step 2: induction step.}
Suppose strategies $\sigma_{t+1},\dots,\sigma_T$ and continuation values $V_{i,t+1}^\sigma,\dots,V_{i,T}^\sigma$ have already been defined. Fix a current state $x=(R,z)$. Consider  the date-$t$ continuation-game payoff defined in \eqref{eq:continuation-game-payoff}, 
\begin{equation}
g_{i,t}(x,a)
:=
u_i(a;x)
+
\delta \sum_{z' \in \mathcal{Z}}
P(z,z')V_{i,t+1}^\sigma\big((R'(x,a),z')\big).
\end{equation}
Because $\mathcal{R}\times\mathcal{Z}$ is finite, the continuation values are finite numbers. Since $R'(x,a)$ is continuous in $a$ and $u_i(a;x)$ is continuous in $a$, the payoff $g_{i,t}(x,\cdot)$ is continuous on the compact action space $\mathcal{A}(x)^N$. Hence the date-$t$ continuation game at state $x$ admits a mixed Nash equilibrium. Select one such equilibrium and denote it by $\sigma_t(\cdot\mid x)$. Then define $V_{i,t}^\sigma(x)$ by \Cref{eq:dynamic-value}.

\emph{Step 3: verification.}
Repeating Step 2 for $t=T-1,T-2,\dots,0$ constructs a full strategy profile $\sigma$. By construction, at every date and state, $\sigma_t(\cdot\mid x)$ is a mixed Nash equilibrium of the continuation game generated by the already constructed continuation values. Therefore the one-shot deviation inequalities in \Cref{eq:mpe-ineq} hold at every date and state. Hence $\sigma$ is a Markov perfect equilibrium.
\end{proof}

\subsection{Proof of Proposition~\ref{prop:mpe-uniqueness}}
\begin{proof}
At date $T$, the claim is immediate because the continuation game is static and, by hypothesis, has a unique symmetric pure Nash equilibrium at each state. Assume recursively that the continuation from dates $t+1,\dots,T$ onward is uniquely pinned down by the pure actions $a_{t+1}^*,\dots,a_T^*$. Then the continuation values appearing in \Cref{eq:continuation-game-payoff} are uniquely determined. By hypothesis, the date-$t$ continuation game therefore has a unique symmetric pure Nash equilibrium action $a_t^*(x)$ at each state $x$. Proceeding backward to $t=0$ yields a unique symmetric pure Markov perfect equilibrium.
\end{proof}

\subsection{Proof of Proposition~\ref{prop:dynamic-foc}}
\begin{proof}
Fix date $t<T$, state $(R,\price,\xi)$, and let all validators other than $i$ choose the common equilibrium action $a=s/N$. Let \[A_{-i}:=\frac{N-1}{N}s.\] If validator $i$ deviates to $a_i$, its total current-plus-continuation payoff is
\begin{align*}
&\Psi_i(a_i;s,R,\price,\xi)
:=\,
(1-\theta)\xi \cdot a_i
-
(1-\theta)\beta \cdot a_i(a_i+A_{-i})
+
\frac{\rho  \price \cdot R\,a_i}{a_i+A_{-i}}
-
\frac{\kappa}{2}a_i^2 \\
&\,
+
\delta\,
\E\!\left[
V_{t+1}\!\left(
R_i'(a_i;s,R,\price,\xi),
\price_{t+1},
\xi_{t+1}
\right)
\middle|\,
\price_t=q,\ \xi_t=\xi
\right]\ ,
\end{align*}
where 
\[
R_i'(a_i;s,R,\price,\xi)
:=
(1-\rho)R+\theta \frac{(\xi-\beta(a_i+A_{-i}))(a_i+A_{-i})}{q}\ .
\]
Because the equilibrium is interior, the first-order condition for optimality is
\[
\left.\frac{\partial \Psi_i}{\partial a_i}\right|_{a_i=s/N}=0.
\]
Differentiating the current-payoff terms gives exactly the derivative computed in \Cref{eq:br-first-deriv}. Differentiating the continuation term by the chain rule yields
\[
\delta
\E\!\left[
\partial_R V_{t+1}(R',\price_{t+1},\xi_{t+1})
\cdot
\frac{\theta}{\price}\big(\xi-\beta(A_{-i}+2a_i)\big)
\middle|\,
\substack{
\price_t=\price,\\
\xi_t=\xi
}
\right],
\]
where
\[
R'=R_i'(a_i;s,R,\price,\xi).
\]
Evaluating at the symmetric profile $a_i=s/N$ gives $A_{-i}=(N-1)s/N$ and hence
\[
\xi-\beta(A_{-i}+2a_i)
=
\xi-\frac{N+1}{N}\beta \cdot s.
\]
The first-order condition therefore becomes
\begin{align*}
&(1-\theta)\xi
-
\frac{N+1}{N}(1-\theta)\beta \cdot s
-
\frac{\kappa}{N}s
+
\frac{N-1}{N}\frac{\rho  \price \cdot R}{s} \\
& \qquad +
\frac{\delta \theta}{\price}
\left(\xi-\frac{N+1}{N}\beta \cdot s\right)
M_{t+1}(R,\price,\xi;s)
=
0.
\end{align*}
Substituting the definitions of $A_N$ and $B_N$ \eqref{eq:A-B} yields \Cref{eq:dynamic-foc}.
\end{proof}

\subsection{Proof of Proposition~\ref{prop:markov-runway}}
\begin{proof}
The first statement is simply the definition of the security set generated by the equilibrium policy. For the second, monotonicity in $R$ implies that for each fixed $(\price,\xi)$ the set
\[
\{R\ge 0: s_t^{\mathrm{M}}(R,\price,\xi)\ge \underline{s}\}
\]
is either empty or a ray of the form $[\SR_t^{\mathrm{M}}(\price,\xi;\underline{s}),\infty)$. This gives the threshold representation.
\end{proof}

\subsection{Proof of Theorem~\ref{thm:parametric-failure}}
\begin{proof}
For each $t\ge 1$, \Cref{eq:parametric-q} implies
\[
\log \price_t \sim N\!\left(\log \price_0+\left(\mu_{\price}-\frac{\sigma_{\price}^2}{2}\right)t,\sigma_{\price}^2 t\right),
\]
so
\begin{align*}
\Prob\!\left(\price_t<\underline{\price}_t(z_{\price})\right)
&=
\Prob\!\left(
\frac{\log \price_t-\left(\log \price_0+\left(\mu_{\price}-\frac{\sigma_{\price}^2}{2}\right)t\right)}{\sigma_{\price}\sqrt{t}}<-z_{\price}
\right)\\
&=
\Phi(-z_{\price}).
\end{align*}
Likewise, \Cref{eq:parametric-xi} implies
\[
\Prob\!\left(\xi_t<\underline{\xi}_t(z_{\xi})\right)=\Phi(-z_{\xi}).
\]
Define
\[
E_T(z_{\price},z_{\xi})
:=
\{\price_t\ge \underline{\price}_t(z_{\price}),\ \xi_t\ge \underline{\xi}_t(z_{\xi})\text{ for all } t=0,\dots,T\}.
\]
By the union bound,
\begin{align*}
\Prob(E_T(z_{\price},z_{\xi})^c)
&\le
\sum_{t=1}^T \Prob\!\left(\price_t<\underline{\price}_t(z_{\price})\right)
+
\sum_{t=1}^T \Prob\!\left(\xi_t<\underline{\xi}_t(z_{\xi})\right) \\
&=
T\big(\Phi(-z_{\price})+\Phi(-z_{\xi})\big).
\end{align*}
Hence
\[
\Prob(E_T(z_{\price},z_{\xi}))
\ge
1-T\big(\Phi(-z_{\price})+\Phi(-z_{\xi})\big).
\]
The hypothesis \Cref{eq:parametric-stress} is exactly the deterministic lower-envelope condition from \Cref{thm:stress-test} with envelopes \Cref{eq:parametric-lower-q}--\Cref{eq:parametric-lower-xi}. Therefore $\tau_{\mathrm{fail}}>T$ on the event $E_T(z_{\price},z_{\xi})$, and \Cref{eq:parametric-failure-bound} follows. The final statement is immediate from the choice of quantiles.
\end{proof}

\subsection{Proof of Proposition~\ref{prop:parametric-handoff}}
\begin{proof}
By \Cref{cor:fee-only}, fee-only feasibility holds exactly when $(1-\theta)\xi_t\ge A_N\underline{s}$, which is equivalent to $\xi_t\ge \xi^{\mathrm{FO}}$. Since
\[
\frac{A_N}{1-\theta}
=
\frac{\kappa}{N(1-\theta)}+\frac{N+1}{N}\beta
>
\beta,
\]
the condition $\xi_t\ge \xi^{\mathrm{FO}}$ implies $\underline{s}\le \xi_t/\beta$, so the capacity condition in \Cref{cor:fee-only} is automatic. This proves the identity for $\tau_{\mathrm{hand}}$.

Let
\[
Y_t:=\log \xi_t = y_0+\mu_{\xi} t+\sigma_{\xi} \sum_{k=1}^t \eta_k.
\]
Because $(\eta_t)$ is i.i.d. with mean zero, the strong law of large numbers implies
\[
\frac{1}{t}\sum_{k=1}^t \eta_k \to 0
\qquad \text{almost surely.}
\]
Hence
\[
\frac{Y_t}{t}\to \mu_{\xi}>0
\qquad \text{almost surely,}
\]
so $Y_t\to \infty$ almost surely. Therefore $\tau_{\mathrm{hand}}<\infty$ almost surely.

For part (ii), the event $\{\xi_T\ge \xi^{\mathrm{FO}}\}$ is contained in $\{\tau_{\mathrm{hand}}\le T\}$. Since
\[
Y_T\sim N(y_0+\mu_{\xi} T,\sigma_{\xi}^2T),
\]
we obtain
\[
\Prob(\tau_{\mathrm{hand}}\le T)
\ge
\Prob(Y_T\ge y^{\mathrm{FO}})
=
1-\Phi\!\left(\frac{y^{\mathrm{FO}}-y_0-\mu_{\xi} T}{\sigma_{\xi}\sqrt{T}}\right).
\]

For the upper bound in part (iii), the tail-sum formula gives
\[
\E[\tau_{\mathrm{hand}}]
=
\sum_{t=0}^{\infty}\Prob(\tau_{\mathrm{hand}}>t).
\]
If $\tau_{\mathrm{hand}}>t$, then $Y_t<y^{\mathrm{FO}}$. Therefore
\[
\Prob(\tau_{\mathrm{hand}}>t)
\le
\Prob(Y_t<y^{\mathrm{FO}})
=
\Phi\!\left(\frac{y^{\mathrm{FO}}-y_0-\mu_{\xi} t}{\sigma_{\xi}\sqrt{t}}\right)
\qquad \text{for } t\ge 1.
\]
This yields the stated upper bound. Because the argument of $\Phi$ is asymptotically of order $-\sqrt{t}$, the Gaussian tail decays exponentially in $t$, so the series converges and $\E[\tau_{\mathrm{hand}}]<\infty$.

Now define the i.i.d. increments
\[
X_k:=\mu_{\xi}+\sigma_{\xi} \eta_k,
\qquad
Y_t=y_0+\sum_{k=1}^t X_k.
\]
Since $\E[\tau_{\mathrm{hand}}]<\infty$ and $\E[X_k]=\mu_{\xi}$, Wald's identity implies
\[
\E[Y_{\tau_{\mathrm{hand}}}]
=
y_0+\mu_{\xi} \E[\tau_{\mathrm{hand}}].
\]
Because $Y_{\tau_{\mathrm{hand}}}\ge y^{\mathrm{FO}}$ by definition of the hitting time,
\[
y_0+\mu_{\xi} \E[\tau_{\mathrm{hand}}]
\ge
y^{\mathrm{FO}},
\]
which proves the lower bound in \Cref{eq:parametric-handoff-lower}.

For part (iv), the event
\[
\{\tau_{\mathrm{hand}}\le T\}\cap \{\tau_{\mathrm{fail}}>T\}
\]
is contained in $\{\tau_{\mathrm{hand}}<\tau_{\mathrm{fail}}\}$. Therefore
\begin{align*}
\Prob(\tau_{\mathrm{hand}}<\tau_{\mathrm{fail}})
&\ge
\Prob(\tau_{\mathrm{hand}}\le T,\tau_{\mathrm{fail}}>T) \\
&\ge
1-\Prob(\tau_{\mathrm{hand}}>T)-\Prob(\tau_{\mathrm{fail}}\le T).
\end{align*}
Applying the bounds from part (ii) and \Cref{thm:parametric-failure} gives \Cref{eq:parametric-success-prob}.
\end{proof}

\end{document}